\documentclass[copyright,creativecommons]{eptcs}
\providecommand{\event}{CCA 2010} 
\usepackage{breakurl}
\usepackage{amssymb}
\usepackage[mathscr]{eucal}
\usepackage{amsmath}
\usepackage{amsfonts}

\newtheorem{theorem}{Theorem}[section]
\newtheorem{definition}[theorem]{Definition}\rm
\newtheorem{lemma}[theorem]{Lemma}

\newtheorem{proposition}[theorem]{Proposition}

\newenvironment{proof}{\noindent\bf Proof. \rm}{\mbox{}\hfill $\square$\vspace*{3mm}}
\def\dotminus{\mbox{  \raisebox{0.6ex}{$\  \cdot$}\hspace{-1.2ex}$-\   $}}

\def\IN{\mathbb{N}}
\def\IQ{\mathbb{Q}}
\def\IR{\mathbb{R}}

\title{On the Weak Computability of Continuous Real Functions}
\author{Matthew S. Bauer  and  Xizhong Zheng
\institute{Department of Computer Science and Mathematics\\ Arcadia University\\
Glenside, PA 19038, USA}
\email{\{mbauer, zhengx\}@arcadia.edu}
}
\def\titlerunning{Weakly Computable Functions}
\def\authorrunning{M. Bauer, \& X. Zheng}
\begin{document}
\maketitle

\begin{abstract}
In computable analysis, sequences of rational numbers which effectively converge to a real number $x$ are used as the ($\rho$-) names of $x$. A real number $x$ is computable if it has a computable name, and a real function $f$ is computable if there is a Turing machine $M$ which computes $f$ in the sense  that, $M$ accepts any $\rho$-name of $x$ as input and outputs a $\rho$-name of $f(x)$ for any $x$ in the domain of $f$.  By weakening the effectiveness requirement of the convergence and classifying the converging speeds of rational sequences, several interesting classes of real numbers of ``weak computability" have been introduced in literature, e.g., in addition to the class of computable real numbers (EC), we have the classes of semi-computable (SC), weakly computable (WC), divergence bounded computable (DBC) and computably approximable real numbers (CA). In this paper, we are interested in the weak computability of continuous real functions and try to introduce an analogous classification of weakly computable real functions. We present definitions of these functions by Turing machines as well as by sequences of rational polygons and prove these two definitions are not equivalent. Furthermore, we explore the properties of these functions, and among others, show their closure properties under arithmetic operations and composition.
\end{abstract}

\section{Introduction}

Computability theory begins with a definition of computable functions operating on the set $\Sigma ^*$ of finite strings over a finite alphabet $\Sigma$. A function $f$ defined on $\Sigma^*$ is computable if there is a Turing machine (TM) $M$ which on input $x \in \Sigma ^*$ outputs $f(x) \in \Sigma ^*$ in a finite number of steps \cite{Rog67, So87}. This computability can be transferred from $\Sigma ^*$ to other countable sets by means of coding systems, but is, however, limited in its ability to induce computability on uncountable sets. For this purpose, the framework of TTE (Type-2 Theory of Effectivity, see \cite{Wei00}), is introduced, allowing for computability over sets of a cardinality up to continuum, like the real numbers.

In TTE, a (type-2) Turing machine $M$ can accept inputs from either $\Sigma ^*$ or $\Sigma ^{\omega}$, where $\Sigma ^{\omega}$ is the set of all infinite strings over $\Sigma$. On inputs and outputs from $\Sigma ^*$, the type-2 Turing machine is defined exactly the same as the classical Turing machine. However,  $M$ can output an infinite sequence $p$ if it operates forever and writes the infinite sequence $p$ on its output tape. If the machine halts in a finite number of steps, or does not write an infinite sequence on its output tape, then the machine diverges, i.e. no output, in this case.

In computable analysis, a real function $f$ is called \emph{computable}, if there is a Type-2 Turing machine $M$ that computes it. That is,  $M$ transforms the name of any real number $x$ in the domain of $f$ into a name of $f(x)$. Here the name of a real number $x$  is in principle a sequence $(x_s)$ of rational numbers which effectively converges to $x$ in the sense that $|x- x_s| \le 2^{-s}$ for any $s$. If, in particular, a real number $x$ has a computable name, then it is called a computable real number. For a computable real number $x$, its computable name offers an effective approximation to $x$ with an effective error estimation. As shown by H.~G.~Rice \cite{Ric54}, a real number is computable if and only if it has a computable binary expansion. In fact, as R.~Robinson \cite{Rob51} has pointed out, all classical mathematical definition of real numbers (by Cauchy sequences, Dedekind's cuts, binary or decimal expansions, nested interval sequences, etc.) can be used to define the notion of computable real numbers equivalently. The class of computable real numbers is denoted by EC (for Effectively Computable).

A real number with a non-computable c.e.~set as its binary expansion is not computable. However it can still be approximated from below. We call a real number {\em left computable} if there is an increasing computable sequence of rational numbers converging to $x$. The right computable real numbers can be defined similarly.   We call a real number {\em semi-computable} if it is either left computable or right computable. The class of semi-computable real numbers is denoted by SC while left and right computable real numbers are denoted by LC and RC, respectively. Clearly, a real number is computable if it is both left and right computable and can be approximated from above and below.

As noted in \cite{AWZ00}, there are two left computable real numbers $y$ and $z$ such that their difference $x := y - z$ is no longer semi-computable. Thus the class of weakly computable real numbers, denoted WC, was introduced as the closure of semi-computable real numbers under the operations ``$+$" and ``$-$". Actually, we can define a real number to be {\em weakly computable} if it is the difference of two left computable real numbers. As shown in \cite{AWZ00}, a real number is weakly computable if and only if there is a computable sequence of rational numbers $(x_n)$ which converges to $x$ ``weakly effectively" in the sense that $\sum_{n=0}^{\infty} |x_{n+1} - x_n|$ is bounded. The class of WC numbers is closed under the arithmetic operations and thus forms a field.

In \cite{ZLB06} Zheng, Lu and Bao investigate the class of divergence bounded computable numbers, denoted DBC. A sequence $(x_n)$ converges $h$-bounded effectively for a total function $h: \IN \to \IN$  if there are at most $h(n)$ non-overlapping pairs $(i,j)$ of indices such that $|x_i - x_j| \geq 2^{-n}$ for all $n \in \mathbb{N}$. A real number $x$ is divergence bounded computable if there is a computable sequence of rational numbers $(x_n)$ that converges to $x$ $h$-bounded effectively with an computable function $h$. It is shown in \cite{ZRG05} that the class of DBC reals is the closure of SC or WC real numbers under total computable real functions. Additionally a real number $x$ is {\em computably approximable}, CA, if there is a computable sequence of rational numbers converging to $x$ without any restriction on the convergence.

As mentioned previously, a real function is computable if there is a Turing machine that maps a $\rho$-name of $x \in dom(f)$ to a $\rho$-name of $f(x)$. This implies immediately that all computable functions are continuous. Several other equivalent definitions of computable real functions also have been discussed in \cite{Grz57, Lac55, PR89}. Particularly, analogous to the classic Weierstrass theorem for the continuous function, a computable function $f$ defined on a closed interval can be described as the limit of a computable sequence of rational polygons which converges uniformly effectively. Here a rational polygon is a continuous piecewise linear function which connects a  finite set of rational turning points on a closed interval. Furthermore, a sequence of rational polygons $(pg_n)$ is uniformly effectively convergent means that  $|pg_n(x) - f(x)| < 2^{-n}$ holds for all $n\in\IN$ and $x$ in the closed interval.

There are, however, some very basic real functions that are not computable. For example, any constant function $f(x) = c$ where $c$ is a non-computable constant is not computable. If this constant happens to be a real number of some weak computability mentioned above, it is quite natural to call the constant function ``weakly computable". Therefore, it is meaningful to introduce reasonable definition of such ``weakly computable" real functions. In \cite{WZ00} Weihrauch and Zheng introduce computability on lower semi-continuous and upper semi-continuous real functions. A real function is {\em lower (upper) semi-continuous} if there is an increasing (decreasing) sequence of polygon functions which converges to it, or, equivalently, it can be approximated from below (above). If the monotonic sequences of polygon functions above are computable, then we can define naturally {\em lower (upper) semi-computable} real functions. However, other than these two classes, the properties related to the continuity of a real function become very complicated. On the other hand, most real problems can be modeled satisfactorily only  by continuous functions.  Therefore, we will focus only on the continuous functions defined on a closed interval in this paper. All functions discussed in this paper are continuous, if it is not pointed out otherwise. Some times we do not even mention the word ``continuous" explicitly. Such functions can be approximated by a sequence of rational polygon functions. Instead of effective convergence of the polygon sequences, we consider weaker effectivity of the convergence and hence introduce several versions of ``weakly computable continuous real functions". In addition, we can also w.l.o.g. consider functions only on the closed interval $[0, 1]$.

By classic Weierstrass theorem, any continuous real function defined on the interval $[0,1]$ is the limit of a sequence of polygon functions. If we consider only the computable sequences of rational polygon functions and require the effectivity of convergence in different levels, then we are able to introduce various classes of continuous functions with different versions of ``computablity". For example, if the computable polygon sequence is increasing (decreasing), then the limiting function is lower (upper) semi-computable. Notice that, although we use the same name, this differs slightly from that of \cite{WZ00} because of the additional continuity requirement. If the sequence converges ``weakly effectively" or ``$h$-bounded effectively", then we can define the classes of ``weakly computable" and ``divergence bounded computable" continuous functions in a straightforward way. We investigate the equivalent definitions of these function classes and explore their mathematical properties.

The paper is organized as follows. In section \ref{sec-pre}, we recall some basic notions and results on the computability of real numbers as well as some fundamental properties related to computable functions. Next we discuss some properties related to the lower and upper semi-computable functions in section \ref{sec-sc}. In section \ref{sec-wc}, we introduce two definitions for weakly computable functions and show that these definitions are not equivalent.

\section{Preliminaries}\label{sec-pre}

In this section, we recall definitions and results related to computability of real numbers and real functions which are used in this paper.

The computability related to the natural numbers can be easily defined by Turing machines. For example, a function $f: \IN \to \IN$ is computable if there is a Turing machine $M$ such that, for any input string $1^n$, $M$ will output the string $1^{f(n)}$. Similarly, computable functions on any countable sets can be defined by such kind of simple coding. Notice that a rational sequence $(x_n)$ is practically a function from $\IN$ to $\IQ$, so we can define the computable rational sequences $(x_n)$ just by the computable function $f$ defined by $f(n) := x_n$.

A real number $x$ is called computable if there is a computable sequence $(x_n)$ of rational numbers which converges to $x$ effectively in the sense that $|x-x_n| \le 2^{-n}$ for all $n$. This is equivalent to a computable sequence $(x_n)$ that converges to $x$ such that $|x_n -x_{n+1}| \le 2^{-n}$ for all $n$. The class of all computable real numbers is denoted by EC.

As the limits of computable increasing (decreasing) sequences of rational numbers, the left and right computable real numbers are the first extensions of computable real numbers. Left computable real numbers also are called computably enumerable (c.e., for short). The union of left and right computable real numbers is the class of semi-computable real numbers. The classes of left, right and semi-computable real numbers are denoted by LC, RC, and SC, respectively.

The class of semi-computable real numbers is unfortunately not closed under arithmetical operations. Ambos-spies, Weihrauch and Zheng \cite{AWZ00} then introduce the class of \emph{weakly computable} real numbers as the arithmetical closure of the semi-computable real numbers. Equivalently, a real number $x$ is weakly computable if it is the difference of two left computable real numbers. This is the reason why weakly computable real numbers are also called d-c.e. (difference of c.e. real numbers). More interestingly, it is shown in \cite{AWZ00} that $x$ is weakly computable if and only if there is a computable sequence $(x_n)$ of rational numbers which converges to $x$ such that $\sum_{n = 0} ^{\infty} |x_n - x_{n+1}|$ is finite. In this case, the sequence is convergent weakly effectively. The class of all weakly computable real numbers is denoted by WC, which is obviously a field.

The class WC has a lot of nice mathematical properties. However, it is not closed under the computable operations (i.e. computable total real functions). The closure of WC under total computable real functions is defined as the class of {\em divergence bounded computable} by \cite{ZLB06} and is denoted by DBC. It turns out the the class DBC is also the closure of LC, RC or SC under the total computable real functions. DBC can also be described as the limits of computable sequences $(x_n)$ of rational numbers which converges $h$-bounded effectively for some total computable function $h$. That is, there are at most $h(n)$ non-overlapping pairs $(i,j)$ of indices such that $|x_i - x_j| \geq 2^{-n}$ for $n \in \mathbb{N}$.  Finally, we consider the limits of computable sequences of rational numbers without any restriction on their convergence, then we achieve the class of computably approximable real numbers (the class CA).

The relationship among these number classes in shown below:
\begin{center}
EC = LC $\cap$ RC $\subsetneq$
\begin{tabular}{c}
  LC \\
  RC
\end{tabular}
$\subsetneq$ SC = LC $\cup$ RC $\subsetneq$ WC $\subsetneq$ DBC $\subsetneq$ CA
\end{center}

For every real number class, there is a condition related to the convergence of a sequence $(x_n)$ as it approaches a real number $x$. Each of these convergence conditions induces different mathematical properties. The classes of semi-computable real numbers has the strongest version of (weak) convergence but induces the smallest set of mathematical properties. Left and right computable real numbers are not closed under the arithmetic operations. We can, however, say that a real number $x$ is LC if and only if $-x$ is right computable. Additionally, if $x$ is LC or RC and $y$ is EC, their sum $x+y$  is still LC or RC.

We now shift our focus to computable real functions. In this paper, we focus on two main definitions of computable real functions. The first is due to Pour-El and Richards \cite{PR89}. A real function $f: [0,1] \to [0,1]$ is computable if and only if there is a computable sequence of rational polygons $(pg_n)_{n \in \mathbb{N}}$ that converges to $f(x)$ uniformly effectively in the sense that $|f(x) - pg_n(x)| \leq 2 ^ {-n}$ for all $x \in [0;1]$. In the preceding definition, we can replace the sequence of rational polygons with a sequence of rational polynomials.

A second, and equivalent definition is due to Weihrauch \cite{Wei00}. Let $f: \subseteq \mathbb{R} \to \mathbb{R}$ be a function and $M$ be a Turing machine. Machine $M$ computes $f$ with respect to the representation $\rho$ if, for any $x \in dom(f)$ and any $\rho$-name $p$ of $x$, $M$ with input $p$ outputs a $\rho$-name of $f(x)$. $f$ is computable or $(\rho, \rho)$-computable if there is a TM which computes $f$ with respect to the representation $\rho$. Here $\rho$ can be replaced with any equivalent representation for a real numbers $x$, such as that of nested interval sequences. Given two computable real functions $f,g \subseteq \mathbb{R} \to \mathbb{R}$ the composition $f \circ g$ is computable. Some examples of computable real functions include $f(x) = -x$, $f(x,y) = x + y$, $f(x,y) = x \times y$.

\section{Lower and Upper Semi-Computable Continuous Functions}\label{sec-sc}
In this section, we give a definition for continuous lower semi-computable functions on the interval [0,1] and investigate some basic properties of these functions. Our first definition utilizes a sequence $(pg_n)$ of rational polygon functions  that converges increasingly. Next we present an equivalent definition by means of a Turing machine.  Finally, we discuss the closure properties under the arithmetical operations and composition. The main idea of this section is very similar to that of \cite{WZ00}, and the proofs here could be simplified a lot by applying results there. However, we avoid mentioning various representations in \cite{WZ00} and state the results and proofs in a more straightforward way. The reader may find it easier to understand.

Let's give our formal definition first.

\begin {definition}\rm \label{def-cs-function}
Let $f: [0; 1] \to  \mathbb{R}$ be a continuous function.

(1) $f$ is {\em lower semi-computable} if there is a computable increasing sequence $(pg _n) _{n \in \mathbb{N}}$ of rational polygon functions  such that for all $x \in dom(f)$  we have $pg_n(x) \leq pg _{n+1}(x)$ and $f(x) = \lim _{n \rightarrow \infty} pg _n (x)$. The class of all lower semi-computable functions is denoted by $\mathbb{LSCF}$.

(2) $f$ is upper semi-computable if there is a computable decreasing sequence  $(pg _n) _{n \in \mathbb{N}}$ of rational polygons  such that for all $x \in dom(f)$  we have $pg_{n+1}(x) \leq pg _{n}(x)$ and $f(x) = \lim _{n \rightarrow \infty} pg _n (x)$. The class of all upper semi-computable functions is denoted by $\mathbb{USCF}$.
\end{definition}

It is easy to see that a real function is computable if it is both lower and upper semi-computable. This means that a computable function $f$ can be approximated by two computable sequences of polygons converging to $f$, one from above and the other from below. We also call a function {\em semi-computable} if it is lower or upper semi-computable and denote the class of all semi-computable functions by  $\mathbb{SCF}$. Notice that, although using the same names, our definition of lower and upper semi-computable functions is slightly stronger than that of \cite{WZ00} because of the requirement of continuity here.

Now we can see that the lower and upper computable functions can be described naturally by Turing machines as follows.

\begin{theorem}\label{thm-eq-lcs}
A continuous function $f: [0; 1] \to  [0; 1]$ is lower semi-computable if and only if there is a type-2 Turing machine that maps an $\rho$-name of each real number $x \in dom(f)$ to an increasing sequence of rational numbers converging to $f(x)$.
\end{theorem}
\begin{proof}
``$\Rightarrow$":  Assume that $f: [0,1] \to \IR$ is a lower semi-computable function. Then, by definition,  there is an increasing computable sequence of rational polygons  $(pg _n) _{n \in \IN}$ which converges to $f$. We want to show that there is a type-2 Turing machine $M$ that computes $f$ in the sense that $M$ takes an effectively convergent sequences of rational numbers $(x_s)_{s \in \mathbb{N}}$ converging to $x$ as input, and outputs an increasing sequence of rational numbers $(z_n)_{n \in \mathbb{N}}$ that converges to $f(x)$, for any $x \in [0,1]$.

Notice that, rational polygon function $pg_i$ is a finite combination of linear functions. Thus we can find a computable modulus function $m_i$ of $pg_i$ such that
\begin{eqnarray}\label{cond-modulus}
|x - y| < 2^{-m_i(s)} \Longrightarrow |pg_i(x) - pg_i(y)| < 2^{-s}
\end{eqnarray}
for all $x, y\in [0,1]$ and all $s \in \IN$. Actually, it is not difficult to see that $m_i(s) = a_i\cdot s$ is such a modulus function, where $a_i$ is the maximum slope of all linear parts of $pg_i$. Without loss of generality we can assume that the modulus function $m_i$ is a strictly increasing function.

For any sequence $(x_s)$ of rational numbers which effectively converges to $x$, we can first define a new sequence $(u_s)$ of rational numbers by
\begin{eqnarray}\label{def-u_s}
u_s := x_{m_s(s)}
\end{eqnarray}
From the effective convergence of the sequence $(x_s)$ we have
\begin{eqnarray}\label{cond-u_s}
|u_s -x| = |x_{m_s(s)} -x| \le 2^{-m_s(s)}
\end{eqnarray}
and hence $\lim u_s = x$.
By (\ref{cond-modulus}), this implies that $|pg_s(u_s) -pg_s(x)| \le 2^{-s}$ for all $s$. In other words, we have
$$pg_s(u_s) -2^{-s} \le pg_s(x) \le f(x).$$
Define a sequence $(y_s)$ of rational numbers by $y_s := pg_s(u_s) -2^{-s}$. Then, we have $y_s \le f(x)$ and $\lim y_s = \lim pg_s (u_s) = f(x)$. From the sequence $(y_s)$ we can now define an increasing sequence $(z_s)$ of rational numbers simply by $z_s := \max\{y_t : t \le s\}$. Clearly, the increasing sequence $(z_s)$ converges to $f(x)$.

The procedure from $(x_s)$ to $(z_s)$ described above obviously can be done by a type-2 Turing machine.

$``\Leftarrow"$: Assume now that we have a Turing machine $M$  that transforms any sequence $(x_s)$ of rational numbers converging effectively to  $x$, into an increasing sequence of rational numbers which converges to $f(x)$ for all $ x \in dom(f)$. From this machine, we want to construct an increasing computable sequence of rational polygons $(pg_n)$ that converges to $f(x)$.

The sequence $(pg_s)$ is constructed in stages. At the stage $s$, we define the rational polygon $pg_s$ as follows.

For any rational number $r$ from $[0; 1]$, the sequence $(r_s)$ defined by $r_s =r$ for all $s$ is a sequence of rational numbers that converges effectively to $r$. For this input $(r_s)$, the machine $M$ should output an increasing sequence $(y_s)$ converging to $f(r)$. In particular, $M$ outputs the rational numbers $y_0, y_1, \cdots, y_s$ in finite steps. In this computation, $M$ can read only a finite initial segment of the input $(r_s)$. Let $r_0, r_1, \cdots, r_k$ be the part of $(x_s)$ which is possibly used in the computation of $M((r_s))$. Then define a rational open interval $I_r := (r - 2^{-k}, r + 2^{-k})$.

Now, for any real number $x \in I_r$, choose arbitrarily a sequence $(x_s)$ of rational numbers which converges to $x$ effectively. Consider a new sequence $(u_s)$ of rational numbers defined by
\[u_s = \left\{
\begin{array}{l l}
  r & \quad \mbox{if $s \leq k$}\\
  x_s & \quad \mbox{if $s > k$}\\
  \end{array}
  \right. \]
The sequence $(u_s)$ converges to $x$ effectively as well. Therefore $M((u_s))$ also will output an increasing sequence converging to $f(x)$. Since $(u_s)$ and $(r_s)$ share the same initial segment of length $k$, $M((u_s))$ will output the same initial part $y_0, y_1, \cdots, y_s$ as $M((r_s))$ does. This means that $y_s \le f(x)$ holds for all $x \in I_r$.

The class $\{I_r : r \in [0,1]\}$ forms an open cover of the closed interval $[0,1]$. By the compactness of $[0,1]$, there is a finite subcover. Because we can enumerate all rational numbers $r$ in $[0,1]$ and for each $r$ we can effectively find the interval $I_r$, this subcover can be determined in finite steps by a simple enumeration procedure.

Let $\{I^1, I^2, \cdots, I^t\}$ be such a finite subcover. Notice that each interval $I_i$ combines a rational number $y^i_s$ such that $y^i_s \le f(x)$ for all $x \in I_i$. Based on this finite collection of intervals we can easily construct a rational polygon function $pg_s$ on $[0,1]$ such that $y_s^i \le pg_s(x) \le f(x)$ for all $x \in I^i$ and for all $i=1,2,\cdots t$. This sequence converges to $f$.  From this polygon sequence ($pg_s$) we can easily construct an increasing sequence of polygons $(pg'_s)$ simply be $pg'_s(x) := \max\{ pg_t(x): t \le s\}$.  This increasing sequence ($pg'_s$) also converges to $f(x)$.
\end{proof}

In the following we list some more properties of semi-computable functions.

\begin{proposition}\label{prop-lc_constant function}
The constant function $f(x) = c$ is lower-semi computable if and only if $c$ is a left computable real number.
\end{proposition}
\begin{proof}
It follows immediately from Theorem \ref{thm-eq-lcs}
\end{proof}

Since the class of all semi-computable real numbers is not closed under the arithmetical operations, the class $\mathbb{SCF}$ is also not closed under the arithmetical operations by Proposition \ref{prop-lc_constant function}. On the other hand, it is easy to see that the class $\mathbb{LSCF}$ and $\mathbb{USCF}$ are closed under addition, and $\mathbb{SCF}$ is closed under the arithmetical operations. For the composition, we have

\begin{proposition}
The classes $\mathbb{LSCF}$, $\mathbb{USCF}$ and $\mathbb{SCF}$ are not closed under composition.
\end{proposition}
\begin{proof}
It is shown in \cite{ZRG05} that the class $\rm DBC$ is a proper superset of {\rm WC} and hence ${\rm SC} \subsetneq {\rm DBC}$. On the other hand, {\rm DBC} is the closure of {\rm LC} or {\rm RC} under total computable real functions. This means that there exist a computable function $f$ (which is, of course also semi-computable) and a semi-computable real number $c$ such that $f(c)$ is not semi-computable. This implies that the classes $\mathbb{LSCF}$, $\mathbb{USCF}$ and $\mathbb{SCF}$ are not closed under composition immediately.
\end{proof}

\begin{proposition}
For any lower semi-computable function $f: [0,1] \to \IR$ the maximum of $f$ is a left computable real number.
\end{proposition}
\begin{proof}
Let $f: [0,1] \to \IR$ be a lower semi-computable function. There is a computable increasing sequence $(pg_n)$ of rational polygon functions such that $f(x) =\lim pg_n (x)$. Let $m_n :=\max\{pg_n(x) : x \in [0,1] \}$. For the rational polygon function, $pg_n$ achieves its maximum at a rational number $c_n$, i.e., $m_n = pg_n(c_n)$. Actually, we can effectively find this $c_n$ for each $n$. Thus, $m_n = pg _n (c_n) \leq pg_{n+1}(c_n) \leq \max(pg _{n+1}) = m _{n+1}$. This implies that $(m_n)$ is an increasing computable sequence of rational numbers. It remains only to show that $(m_n)$ converges to $m:= \max\{f(x) : x\in [0,1]\}$.

Since $f$ is a continuous function on $[0,1]$, there is a real number $c \in [0,1]$ such that $f(c) =m$. The sequence $(a_n)$ defined by $a_n :=pg_n(c)$ is an increasing sequence which converges to $m$. Because $m_n \ge a_n$ and $m_n \le m$, we have $\lim m_n =m$. That is, the maximum of $f$ is a left computable real number.
\end{proof}

\section{Weakly Computable Continuous Functions}\label{sec-wc}

In this section, we introduce the class of weakly computable continuous functions. They are defined first as the difference of two lower semi-computable functions. Next, we show that a continuous function is weakly computable if and only if it can be computed by Turing machine which transfers the $\rho$-names of $x$ into a sequences which converge to $f(x)$ weakly effectively. However, we will find that the characterization by sequences of polygon functions seems to induce a slightly stronger version of weakly computable functions.

\begin{definition}\label{defWC}\rm
A real function $f: [0;1] \to \IR$ is called {\em weakly-computable} if there are two lower semi-computable functions $h$ and $g$ such that $f(x) = h(x) - g(x)$. The class of all weakly-computable functions is denoted by $\mathbb{WCF}$.
\end{definition}

By definition, all semi-computable functions are weakly computable while there are weakly computable functions which are not semi-computable. An easy example is the constant function $f(x) =c$ for weakly computable real number $c$ which is not semi-computable. So we have $\mathbb{SCF} \subsetneq \mathbb{WCF}$. The next lemma shows some other basic properties of the class $\mathbb{WCF}$.

\begin{lemma}
\begin{enumerate}
  \item The class $\mathbb{WCF}$ is closed under the arithmetical operations.
  \item The class $\mathbb{WCF}$ is not closed under the composition.
\end{enumerate}
\end{lemma}

\begin{proof}
The proof of (1) follows directly from definition \ref{defWC}.

For the proof of (2), notice that the class WC of weakly computable real numbers is not closed under the total computable function (see \cite{ZRG05}). That is, there is a computable function $f:[0,1] \to \IR$ and a real number $c\in [0,1]$ such that $f(c)$ is not weakly computable. The computable function $f$ as well as the constant function  $g(x) =c$ are weakly computable.  The composition $f \circ g(x)$ = $f(c)$ is not a weakly computable function.
\end{proof}

The weakly computable functions can be naturally described by Turing machines as follows.

\begin{theorem}\label{WC_TM}
A continuous real function $f: [0; 1] \to \IR$ is weakly computable if and only if there is a Turing machine $M$ that transfers any sequence $(x_s)$ which effectively converges to $x \in [0,1]$ into a sequence $(u_s)$ of rational numbers which converges weakly effectively to $f(x)$.
\end{theorem}

\begin{proof}
$``\Rightarrow"$
Let $f:[0,1] \to \IR$ be a weakly computable function. Then there are lower semi-computable functions $g$ and $h$ such that $f=g-h$. By Theorem \ref{thm-eq-lcs}, there are two Turing machines $M_1$ and $M_2$ such that, for any input of sequence $(x_s)$ of rational numbers which converges effectively to $x\in[0,1]$, $M_1$ and $M_2$ will output increasing sequences $(y_s)$ and $(z_s)$ of rational numbers which converge to $g(x)$ and $h(x)$, respectively.  From $M_1$ and $M_2$, it is easy to construct a new Turing machine $M$ which outputs the sequence $(y_s -z_s)$ for the input $(x_s)$.  Let $u_s:= y_s-z_s$ for all $s$. Then the sequence $(u_s)$ converges to $f(x)$. This sequence converges actually weakly effectively because of the following estimation.
\begin{eqnarray*}
\sum_{n=0}^\infty |u_{n+1} -u_n|
&=& \sum_{n=0}^\infty|(y_{n+1} -z_{n+1})-(y_n-z_n)|\\
&=& \sum_{n=0}^\infty|(y_{n+1} -y_n) - (z_{n+1}-z_n)|\\
&\le& \sum_{n=0}^\infty|y_{n+1} -y_n| + \sum_{n=0}^\infty|z_{n+1}-z_n|\\
&=& \sum_{n=0}^\infty(y_{n+1} -y_n) + \sum_{n=0}^\infty(z_{n+1}-z_n)\\
&=& \lim_{n\to\infty} y_n -y_0 + \lim_{n\to\infty} z_n -z_0 = g(x) +h(x) -y_0 -z_0.
 \end{eqnarray*}

``$\Leftarrow$": Now assume that we have a Turing machine $M$ that transfers any sequence $(x_s)$ of rational numbers converging effectively to $x$ into a sequence $(u_n)$ of rational numbers which converges weakly effectively to $f(x)$. From this sequence $(u_s)$ we can define two increasing sequences $(y_s)$ and $(z_s)$ of rational numbers by
\begin{eqnarray*}
  y_s := u_0 + \sum^s_{i=0}(u_{i+1} \dotminus u_i)\ \mbox{ and }
  z_s := \sum^s_{i=0}(u_{i} \dotminus u_{i+1}).
\end{eqnarray*}
Since the sequence $(u_s)$ converges weakly effectively, there is a constant $c$ such that $\sum |u_{s+1} -u_s| \le c$. This implies immediately that the sequences $(y_s)$ and $(z_s)$ are bounded above by $u_0+c$ and $c$, respectively. Hence they are both convergent sequences.  Let $g(x):= \lim_{s\to\infty} y_s$ and $h(x):= \lim_{s\to\infty} z_s$ be their limits. This means that we can construct two Turing machines $M_1$ and $M_2$ from $M$ which transfer the sequence $(x_s)$ to the increasing sequences $(y_s)$ and $(z_s)$, respectively. By Theorem \ref{thm-eq-lcs}, the functions $g$ and $h$ computed by $M_1$ and $M_2$ are lower semi-computable and hence $f =g-h$ is weakly computable.
\end{proof}

We are also interested in the characterization of weakly computable functions by computable sequences of rational polygon functions. From Theorem \ref{WC_TM} we can easily get an equivalent characterization of weakly computable functions by computable sequences of rational polygons which converges point-wise and weakly effectively. However, the situation is slightly different if we consider the uniformly convergent polygon sequences.

In the case of real numbers, $x$ is weakly computable if there is a computable sequence $(x_s)$ of rational numbers which converges weakly effectively to $x$, that is $\sum_{s=0}^\infty |x_{s+1} -x_s| \le c$ for some constant $c$. For different sequences, the constant $c$ may change. However, by deleting some initial terms, we actually can replace the constant $c$ by 1 and still have the equivalent definition of weakly computable real numbers.

For functions, the situation is completely different. If we try to characterize the weakly computable function $f:[0,1]\to \IR$ by sequences of rational polygons,  then we need a computable sequence $(pg_s)$ of rational polygons converging to $f$ such that $\sum_{s=0}^\infty |pg_{s+1}(x), pg_s(x)| \le c$ for some constant $c$. This number $c$ depends on the argument $x$. We can naturally require a uniform constant $c$ that works for all $x$ in the interval $[0,1]$. This is equivalently to require that $\sum_{s=0}^\infty d(pg_{s+1}, pg_s) \le c$, where $d(f,g):=\max\{|f(x)-g(x)|: 0\le x \le 1\}$ is the {\em distance} between the functions $f$ and $g$. Our next theorem shows that this kind of polygon sequence characterization is equivalent to the uniformity requirement of the constant $c$ in the weakly convergent outputs of the Turing machines. In this case, the uniform constant $c$ can be simply replaced by the constant number 1.

\begin{theorem}\label{thm-wcf-uniform}
Let $f: [0,1] \to [0; 1]$ be a continuous real function. The following are equivalent.
\begin{enumerate}
  \item There is a computable sequence $(pg_s)$ of rational polygon functions which converges to $f$ weakly effectively in the sense that $\sum_{s=0}^\infty d(pg_{s+1}, pg_s) \le c$ for some constant $c$.
  \item There is a Turing machine $M$ which transfers any sequence of rational numbers effectively converging to $x\in [0,1]$ into a rational sequence $(y_s)$ converging to $f(x)$ and $\sum_{s=0}^\infty |y_{s+1} -y_s| \le 1$.
\end{enumerate}
\end{theorem}

\begin{proof}
``$\Rightarrow$":
Suppose that  $(pg_s)$ is a computable sequence of rational polygon functions which converges to $f$ such that $\sum_{s=0}^\infty d(pg_{s+1}, pg_s) \le c$ for some constant $c$. By deleting some initial terms of the sequence, we can reduce the sum $\sum_{s=0}^\infty d(pg_{s+1}, pg_s)$. So, without loss of generality, we can suppose that $\sum_{s=0}^\infty d(pg_{s+1}, pg_s) \le 1/2$. From the computable sequence $(pg_s)$ of rational polygon functions, we can find a computable uniform modulus function $m:\IN^2 \to \IN$ such that
\begin{eqnarray}\label{cond-unif-modulus}
|x-y| \le 2^{-m(i,s)} \Longrightarrow |pg_i(x) -pg_i(y)| \le 2^{-(s+3)}
\end{eqnarray}\label{cond-x_s-m}
for all $x, y \in [0,1]$. We can assume that $m(i,s)$ is increasing for both $i$ and $s$. Then we have
\begin{eqnarray}
|x -x_{m(s,s)}| \le 2^{-m(s,s)} \qquad \mbox{and}\qquad
|pg_s(x) - pg_s(x_{m(s,s)})| \le 2^{-(s+3)}.
\end{eqnarray}
if the sequence $(x_s)$ converges to $x$ effectively.

Now we can construct a Turing machine $M$ as follows: for any input sequence $(x_s)$ of rational numbers, $M$ simply outputs the rational sequence $(y_s)$ defined by $y_s := pg_s(x_{m(s,s)})$ for all $s$. If the sequence $(x_s)$ converges to $x$ effectively, then $\lim_{s\to\infty} y_s = \lim_{s\to\infty} pg_s(x_{m(s,s)}) = \lim_{s\to\infty} pg_s(x) =f(x)$. It remains only to show that $\sum_{s=0}^\infty |y_{s+1} -y_s| \le 1$. From the following estimation,
\begin{eqnarray*}
&& |y_{s+1} -y_s| = |pg_{s+1}(x_{m(s+1, s+1)}) - pg_s(x_{m(s,s)})|\\
&\le&  |pg_{s+1}(x_{m(s+1, s+1)})-pg_{s+1}(x)|+
        |pg_{s+1}(x)- pg_s(x)|+
        |pg_s(x) - pg_s(x_{m(s,s)})|\\
&\le& 2^{-(s+3)} + |pg_{s+1}(x)- pg_s(x)| + 2^{-(s+3)}\\
&\le& 2^{-(s+2)} + d(pg_{s+1}, pg_s),
\end{eqnarray*}
we have $\sum_{s=0}^\infty |y_{s+1} -y_s| \le 2^{-1} + \sum_{s=0}^\infty d(pg_{s+1}, pg_s) \le 1/2 +1/2 = 1$.

``$\Leftarrow$":
Suppose that $M$ is a Turing machine which transfers any sequence of rational numbers effectively  converging to $x\in [0,1]$ into a rational sequence $(y_s)$ converging to $f(x)$ and $\sum_{s=0}^\infty |y_{s+1} -y_s| \le 1$. In order to reflect the dependence of this sequence to the real number $x$, we will denote $y_s$ more suitably by $y_s(x)$. By the compactness of the closed interval $[0,1]$ and the similar construction in the proof of Theorem \ref{thm-eq-lcs}, we can construct a computable sequence $(pg_s)$ of rational polygon functions such that
\begin{eqnarray*}
|y_s(x) -pg_s(x)| \le 2^{-s}
\end{eqnarray*}
for any $x\in [0,1]$ and any $s$.

Notice that, as continuous functions, polygon functions $pg_s$ and $pg_{s+1}$ achieve their maximum distance at some point $c\in[0,1]$, that is $d(pg_{s+1}, pg_s) = |pg_{s+1}(c) -pg_s(c)|$. Thus we have the the following inequality
\begin{eqnarray*}
&& d(pg_{s+1}, pg_s) = |pg_{s+1}(c) -pg_s(c)|\\
&\le& |pg_{s+1}(c) -y_{s+1}(c)| +|y_{s+1}(c) -y_{s}(c)| +|y_{s}(c) -pg_s(c)|\\
&\le& 2^{-s+1} +|y_{s+1}(c) -y_{s}(c)|
\end{eqnarray*}
Therefore, can conclude that $\sum_{s=0}^\infty d(pg_{s+1}, pg_s) \le 3 + \sum_{s=0}^\infty |y_{s+1}(c) -y_s(c)| \le 4$.
\end{proof}

The functions characterized in theorem \ref{thm-wcf-uniform} require the existance of a uniform constant bounding the weakly effective convergence of all sequences $(y_s)$ converging to $f(x)$ output by a Turing machine. Therefore, we can naturally call them {\em uniformly weakly computable functions}, denoted $\mathbb{UWCF}$. The next theorem proves that the class of uniformly weakly computable functions is smaller than the class of weak computable functions.

\begin{theorem}
There exists a weakly computable function $f$ such that $f$ is not uniformly weakly computable. That is, if a computable sequence $(pg_s)$ of rational polygons converges to $f$, then the convergence is not weakly effective, i.e.,  the sum $\sum_{s=0}^\infty d(pg_s, pg_{s+1})$ is infinite.
\end{theorem}

\begin{proof}
We will construct a weakly computable function $f:[0,1] \to \IR$ as the limit of a computable sequence $(f_s)$ of rational polygons. To be  weakly computable, the function $f$ must be the difference of two lower semi-computable functions $g$ and $h$, that is, $f = g-h$. Thus, we need only to construct two increasing computable sequences $(g_s)$ and $(h_s)$ of rational polygons which converge to $g$ and $h$, respectively, and then let $f_s =g_s -h_s$ for all $s$.

In addition, $f$ must be different from the limits of any computable sequences of rational polygons which converge uniformly and weakly effectively. That is, if $(pg_s)$ is a computable sequence of rational polygons such that $\sum_{s=0}^\infty d(pg_{s+1}, pg_s) \le c$ for some constant $c$, then $\lim pg_s \neq f$. By just deleting some initial terms of the sequence, we can find  another computable sequence $(pg'_s)$ which has the same limit such that  $\sum_{s=0}^\infty d(pg'_{s+1}, pg'_s)\le 1$. Therefore, without loss of generality,  we can consider only the computable sequences of rational polygons $(pg_s)$ such that $\sum_{s=0}^\infty d(pg_{s+1}, pg_s)\le 1$. Suppose that  $(M_e)_{e\in \IN}$ is an effective enumeration of all Turing machines which possibly compute a sequence of rational polygons, then it suffices that the constructed function $f$ satisfies all the following requirements:
\begin{eqnarray*}
R_e:\qquad  \mbox{If $M_e$ computes $(pg^e_s)_{s}$ and $\sum_{s=0}^\infty d(pg^e_{s+1}, pg^e_s)\le 1$, then $\lim_{s\to\infty} pg^e_s \neq f$.}
\end{eqnarray*}

To satisfy the requirement $R_e$, we choose a number $x_e :=2^{-e}$ in the interval $[0,1]$ and then try to define $f_s(x_e)$ in such a way that $|\lim pg^e_s(x_e) - f(x_e)| \ge 2^{-e}$, if the sequence $(pg^e_s)_s$ converges weakly effectively. That is, $x_e$ is a witness of the requirement of $R_e$.

The sequence $(f_s)$ for $f_s := g_s -h_s$ is constructed in stages.

Initially, at the stage $s=0$, let $f_0(x) = g_0(x) =h_0(x) =0$ for all
$x$ in the interval $[0,1]$.

At stage $s+1$, suppose that we have defined the finite sequence $(f_t)_{t\le s}$, as well as $(g_t)_{t\le s}$ and $(h_t)_{t\le s}$. Both $(g_t)_{t\le s}$ and $(h_t)_{t\le s}$  are increasing sequences of rational polygons and $f_t =g_t- h_t$ for all $t\le s$. Then, we simulate the computations of all Turing machines $M_e$ up to $s$ steps for $e\le s+1$ to determine if $M_e$ outputs an initial segment of a rational polygon sequence, i.e., if it outputs a sequence $(A_i)_{i\le n_{e_s}}$, for some natural number $n_{e_s}$, of finite sets of rational points which determines a rational polygon on the interval $[0,1]$. If it is not the case, then go directly to the next stage without doing anything. Otherwise, suppose that $(pg^e_i)_{i\le n_{e_s}}$ is the the finite sequence of rational polygon functions which is computed by the Turing machine $M_e$ up to stage $s$. If $\sum_{i=0}^{n_{e_s}} |pg^e_{i+1}(x_e) -pg^e_i(x_e)| >1$ or $|f_s(x_e) -pg^e_{n_{e_s}}(x_e)| > 2^{-e}$, then we needn't do anything. Otherwise, redefine $g_{s+1}(x_e) = g_s(x_e) + 2^{-e}$ if  $f_s(x_e) \le pg^e_{n_{e_s}}(x_e)+ 2^{-e}$, and redefine $h_{s+1}(x_e) = h_s(x_e) + 2^{-e}$ if  $f_s(x_e) \ge pg^e_{n_{e_s}}(x_e)+ 2^{-e}$. Finally, redefine the rational polygons $g_{s+1}$, $h_{s+1}$ and $f_{s+1}$ accordingly. Particularly, if $e$  is the largest index such that $g_{s+1}(x_e)$  was redefined, we should define $g_{s+1}(x) =g_{s+1}(x_e)$ for all $x$ in the interval $[0, x_e]$. $h_{s+1}$ should be defined in a similar way.

By the construction above, $g_{s+1}$ and $h_{s+1}$ are obviously increasing computable sequences of rational polygon functions. Because $g_s(x_e)$ or $h_s(x_e)$ will be increased by $2^{-e}$ only if the polygon sequence $(pg^e_i)$ computed by $M_e$ makes a jump larger than $2^{-e}$ on the argument $x_e$ and because of the condition $\sum_{i=0}^{n_{e_s}} |pg^e_{i+1}(x_e) -pg^e_i(x_e)| \le 1$, there are at most $2^e$ possible changes of $g_s$ and $h_s$ on the point $x_e$. Therefore, both sequences $(g_{s})$ and $(h_{s})$ are bounded above and hence converge.

On the other hand, if $k$ is the limit of a computable sequence of rational polygons which converges uniformly and weakly effectively, then there is a Turing machine $M_e$ which computes a sequence $(pg^e_s)_s$ of rational polygons such that $\sum_{s=0}^\infty d(pg^e_{s+1}, pg^e_s) \le 1$ and $\lim_{s\to\infty} pg^e_s =k$. Then, by the construction, we have $|f_{s+1}(x_e) -pg^e_{n_{e_s}}(x_e)| \ge 2^{e+1}$ for any $s$. This implies that, $f(x_e) = \lim f_s(x_e) \neq \lim pg^e_s(x_e) =k(x_e)$. Therefore, the function $f$ is weakly computable but not uniformly weakly computable.
\end{proof}

\end{document}